\documentclass[11pt,onecolumn]{article}

\usepackage[utf8]{inputenc}
\usepackage{amssymb}
\usepackage{amsthm}
\usepackage{amsmath}
\usepackage{array}
\usepackage{booktabs}
\usepackage{graphicx}
\usepackage{empheq}
\usepackage{mathtools}
\usepackage{makecell}
\usepackage{stmaryrd}
\usepackage[vlined,linesnumbered,ruled,resetcount]{algorithm2e}
\usepackage{hyperref}
\usepackage{dsfont}
\usepackage{tikz}
\usepackage{url}
\usepackage{adjustbox}
\usepackage{enumitem}
\usepackage{cleveref}

\usepackage[left=1.0in,top=1.0in,right=1.05in,bottom=1.0in,nohead]{geometry}

\theoremstyle{plain}
\newtheorem{theorem}{Theorem}[section]

\newtheorem{definition}[theorem]{Definition}
\newtheorem{lemma}[theorem]{Lemma}

\def\slack{{\tt{slack}}}
\def\calP{{\cal P}}
\def\r{{\mathfrak r}}
\def\calB{{\cal B}}
\def\calC{{\cal C}}
\def\calG{{\cal G}}
\def\InvA{{\ensuremath{({\tt I}_0)}}}
\def\InvB{{\ensuremath{({\tt I}_1)}}}
\newcommand{\eqdef}{{\stackrel{\mbox{\tiny \tt ~def~}}{=}}}
\def\myparagraph#1{\vspace{2pt}\noindent{\bf #1~~}}

\title{Blossom VI: A Practical Minimum Weight Perfect Matching Algorithm}
\author{Pavel Arkhipov \hspace{30pt} Vladimir Kolmogorov \\ \normalsize Institute of Science and Technology Austria \\ {\normalsize\tt $\{$pavel.arkhipov,vnk$\}$@ist.ac.at}}

\begin{document}

\maketitle

\begin{abstract}
Minimum weight perfect matching is a fundamental problem in combinatorial optimization. Since 2009, Blossom V has been the leading practical implementation. We present Blossom VI, a practical algorithm that improves upon Blossom V. We test the performance of our implementation on nine benchmark families with graphs containing up to roughly 10 million edges. Blossom VI is significantly faster on the hardest tested instances, while being at most a factor of two slower on families where Blossom V already scales nearly linearly.

Blossom VI is a primal-dual solver whose primal phase computes a maximum-cardinality matching in the zero-slack subgraph using cherry trees. At the end of each primal phase, Blossom VI contracts entire cherry blossoms instead of sequences of nested traditional blossoms that would be created by Blossom V. This produces shallower contraction hierarchies and avoids costly cascades of blossom expansions. Our measurements show that reduced blossom depth and fewer expensive expansion operations can explain the observed improvement.
\end{abstract}

\section{Introduction}

We study the minimum weight perfect matching problem. Given a graph $G = (V, E)$ with edge weights $w_e$, one needs to compute a perfect matching of minimum total weight, i.e., a subset of edges $M \subseteq E$ that minimizes the total weight $\sum\limits_{e \in M} w_e$ subject to the constraint that every node has exactly one incident edge from $M$.

Minimum weight perfect matching arises as a subroutine in several applications. For example, it is used in Christofides' approximation
algorithm for the metric traveling salesperson problem
\cite{christofides2022worst}, in polynomial algorithms for maximum cut in
planar graphs \cite{hadlock_planar_maxcut}, and in quantum error correction
\cite{higgott2022pymatching, higgott2025sparse, fowler2013minimum}. These applications motivate implementations
that remain efficient on large, structured instances.

The existence of a polynomial-time algorithm for solving the minimum weight perfect matching problem is due to Edmonds \cite{edmonds_lp, edmonds_algo}. Since then, the problem has inspired a long line of research. On the theoretical side, the complexity of the algorithm has been improving through \cite{mwpm_theory_gabow, mwpm_theory_gabow2, mwpm_theory_gabow3, mwpm_theory_gabow4, mwpm_theory_gabow5, mwpm_theory_huang}. However, algorithms with good theoretical complexity are not necessarily the fastest in practice. Starting with the Blossom I implementation \cite{blossom_1}, notable implementations include Blossom IV \cite{blossom_4}, the Mehlhorn-Sch\"afer implementation \cite{mehlhorn_schafer} and Blossom V \cite{blossom_v}. In this paper, we focus on the practical implementation of the minimum weight perfect matching algorithm. 

Blossom V, to our knowledge, remains the leading general-purpose practical implementation for solving the minimum weight perfect matching problem. For many natural graph families, Blossom V runs in almost linear time, even for large graphs with millions of edges. However, it can experience a slowdown on difficult structured graphs. For example, the original paper reports this behavior for the heaviest TSP instances (\cite{blossom_v}, Table 9). In this work, we present an implementation that aims to outperform Blossom V on the hard instances. Following the tradition of naming the minimum weight perfect matching algorithms, we call our algorithm Blossom VI.

\subsection{Our contributions}

Our main contribution is the development of Blossom VI, a practical algorithm for the
minimum weight perfect matching problem. 
% We build our work upon Blossom V algorithm and the cherry-tree algorithm \cite{cherry_paper} that solves the unweighted maximum-cardinality matching problem. A version of the cherry-tree algorithm is used in every primal phase of our implementation. 
Like other primal-dual matching
algorithms, Blossom VI alternates between two types of steps. A
\emph{primal phase} improves the current matching using edges whose dual
constraints are tight, and a \emph{dual phase} changes the dual variables
so that new primal operations may become available.

Our main observation is that each primal phase can be viewed as an
unweighted maximum-cardinality matching problem. During such a phase, the dual variables
remain fixed, and the algorithm searches for matching augmentations using the edges that are currently tight. The goal of the
primal phase is therefore to find a maximum-cardinality matching using
only these edges, starting from the current matching. This viewpoint
allows us to use a variant of an efficient unweighted maximum matching algorithm based on cherry
trees \cite{cherry_paper}.

The main algorithmic difference from Blossom V is when and how odd vertex sets are contracted into supernodes. A
traditional implementation may create several nested blossoms during one
primal phase; the intermediate blossoms initially have dual
variable zero; nevertheless, they must be represented explicitly and may
later have to be expanded one after another. Blossom VI keeps this intermediate structure implicit during the primal
phase. We contract
relevant odd sets at the end of each primal phase. Thus, several nested blossom contractions that would
otherwise be created during one primal phase are replaced by one
larger contraction. The resulting blossom hierarchy is shallower, fewer
contracted objects have to be initialized, and costly cascades of blossom
expansions are avoided. 

We implement Blossom VI in C++ and compare it experimentally with the state-of-the-art practical implementation, Blossom V.
On instance families for which Blossom V shows superlinear running
time, Blossom VI is substantially faster. On families that Blossom V already
solves in nearly linear time, Blossom VI remains competitive but can be
up to a factor of two slower. Our measurements indicate that the largest
speedups coincide with shallower blossom hierarchies and
with avoiding expensive sequences of expansion operations.

Our code is available on GitHub: \url{https://github.com/Paul566/blossom-vi}.

\section{Preliminaries}

In this section, we give a brief overview of the problem and the primal-dual technique, as well as a recap of some ideas from the cherry-tree algorithm for unweighted maximum matching.

We start with notation. For a graph $G$, we use $V(G)$ and $E(G)$ to denote the sets of nodes and edges of $G$, respectively. If $G$ is directed, we write $A(G)$ for the set of arcs. $\delta(v)$ denotes the set of incident edges of a node $v$. For a set of nodes $S$, we use $\delta(S)$ for the edge boundary of $S$. $E(S)$ denotes the set of edges with both endpoints in $S$. The input graph has $n$ nodes and $m$ edges.

\subsection{Linear programming formulation and the primal-dual approach}

We give a brief overview of the linear programming formulation and the primal-dual method of solving the problem. It was originally studied by Edmonds \cite{edmonds_lp, edmonds_algo} and is classical. The minimum weight perfect matching admits a linear programming formulation. As a first attempt, one can introduce Boolean variables $x_e \in \{0, 1\}$ and state the problem as an \emph{integer} linear program: minimize $\sum\limits_{e \in E} w_e x_e$ subject to the constraints $\sum\limits_{e \in \delta(v)} x_e = 1$. However, this would be an ILP that is not a priori efficiently solvable. If one drops the integrality constraints, the solution might become non-integral. The key structural theorem is that one can enforce integrality if one adds a new family of constraints, called the \emph{blossom constraints}. Let $\mathcal{P}_\text{odd}$ be the family of node sets of odd cardinality that have more than one node. Then the minimum weight perfect matching problem admits the following primal and dual LP formulations.

\begin{equation}
    \begin{aligned}
    \text{(Primal)} \qquad 
    \min \quad & \sum_{e \in E} w_e\, x_e \\
    \text{s.t.} \quad 
    & \sum_{e \in \delta(v)} x_e = 1 
    && \forall v \in V, \\
    & \sum_{e \in \delta(S)} x_e \geq 1
    && \forall S \in \mathcal{P}_\text{odd}, \\
    & x_e \ge 0
    && \forall e \in E.
    \end{aligned}
\end{equation}
\begin{equation}
    \begin{aligned}
    \text{(Dual)} \qquad 
    \max \quad & \sum_{v \in V} y_v \;+\; \sum_{S \in \mathcal{P}_\text{odd}} y_S \\
    \text{s.t.} \quad 
    & \slack(e) \geq 0
    && \forall e \in E, \\
    & y_S \geq 0
    && \forall S \in \mathcal{P}_\text{odd}.
    \end{aligned}
\end{equation}
For an edge $e = (u, v)$, the slack is defined as 
\begin{equation}
    \label{eq_slack_def}
    \slack(e) := w_e - y_u - y_v - \sum\limits_{S \in \mathcal{P}_\text{odd} \; : \; e \in \delta(S)} y_S.
\end{equation}
The complementary slackness conditions are:
\begin{equation}
    \begin{cases}
        x_e \cdot \slack(e) = 0 \\
        y_S (x(\delta(S)) - 1) = 0
    \end{cases}
\end{equation}

The minimum weight perfect matching problem can be solved with a classical primal-dual method. At a high level, the method maintains a feasible dual solution $y$ and a not necessarily feasible integral primal solution $x$. 
The primal solution $x$ corresponds to a matching $M$ that is non-perfect in general: $x$ satisfies $x_e \geq 0$ and $\sum\limits_{e \in \delta(v)} x_e \leq 1$, but it might violate the equality constraints $\sum\limits_{e \in \delta(v)} x_e = 1$ and the blossom constraints. 
Furthermore, maintain $x$ and $y$ such that every matched edge has slack zero, and for any $S \in \calP_{\text{odd}}$ with $y_S > 0$, we have $|M \cap E(S)| = (|S| - 1) / 2$. Perform a sequence of updates for $x$ and $y$ until $x$ becomes feasible, and therefore corresponds to a perfect matching. Then, the maintained invariants imply the complementary slackness conditions. The feasibility of $x$, the feasibility of $y$, and the complementary slackness conditions imply that $x$ is optimal. 

The blossom constraints have a direct algorithmic counterpart. Traditional
primal-dual algorithms maintain a laminar family of odd-size node sets, called
\emph{blossoms}. During the primal phase, the matching is improved by \emph{augmentations}, i.e., flipping the matching on an alternating path of tight edges between two unmatched nodes. When the augmenting path search on tight edges discovers an odd alternating
cycle on zero-slack edges, the corresponding blossom is contracted into a single supernode, so
that the search can continue. The
blossom can reconstruct the matching inside itself if an
augmenting path later passes through it, and also stores the dual
variable $y_S$ associated with the contracted set $S$. Blossoms may be nested.
During dual updates their dual variables change; blossoms whose internal
structure must become visible again are expanded. For a more detailed description, see, for example, the book \cite{schrijver}, chapters 25--26.

\subsection{Unweighted maximum matching}
As a subroutine in our implementation, we solve the unweighted maximum matching problem. The objective is to find a matching of maximum cardinality in a given unweighted graph. 

Let $M$ be a matching in an unweighted graph $G$. An
\emph{alternating path} is a path whose edges alternate between edges in
$M$ and edges outside $M$. An alternating path whose two endpoints are
unmatched is called an \emph{augmenting path}. Replacing the matched
edges of such a path by its unmatched edges increases the cardinality of
$M$ by one. A matching is maximum if and only if no
augmenting path exists. 
This idea is used both in the theoretical algorithms
for maximum matching
\cite{maxmatching_theory_gabow,maxmatching_theory_vazirani,
maxmatching_theory_vazirani2,maxmatching_theory_gaussian,
maxmatching_theory_vazirani3}
and in practical implementations such as Boost~\cite{boost},
LEMON~\cite{lemon}, the cherry-tree implementation~\cite{cherry_paper},
and X-Blossom~\cite{x_blossom}.
Finding augmenting paths in general graphs is complicated by odd cycles.
Edmonds' classical solution contracts certain odd vertex sets, called
\emph{blossoms}, and continues the search in the contracted graph
\cite{edmonds_algo}.

Our primal phase builds upon the cherry-tree algorithm
\cite{cherry_paper}. Starting from the current matching, the algorithm
grows several search trees, one rooted at each unmatched
node. A node receives a plus label if the tree contains an
even-length alternating path from that node to its root, and a minus
label if it contains an odd-length such path. An edge joining two plus
nodes in two different trees gives an augmenting path and therefore
allows the matching to be enlarged.

The main distinction from a traditional blossom search is how odd
cycles are dealt with. The cherry-tree algorithm does not contract every odd cycle encountered by the search but maintains \emph{cherry blossoms} that generalize traditional blossoms. A cherry blossom is an odd-cardinality set of nodes that has a special
\emph{receptacle} node, and all other nodes in the cherry blossom have both an
even-length and an odd-length alternating path to the receptacle. This
structure contains enough information to route an augmenting path
through the cherry blossom without representing it as a hierarchy of
contracted odd cycles.

The original cherry-tree algorithm also proposes a recursive
\emph{metagraph} procedure for choosing larger collections of augmentations. The nodes of the metagraph represent cherry trees, and
its edges represent the augmentations currently available between
pairs of trees (hence the name of their paper ``Shrinking trees not blossoms''). Blossom VI does not use this procedure.
The precise definitions of plus- and minus-nodes, receptacles, cherry
trees, and cherry blossoms are given in the next section.

\section{Overview of the algorithm}\label{sec_algo_overview}

At any time, Blossom VI works on a graph obtained from the original graph by contracting some odd node sets into supernodes. The algorithm alternates between primal and dual phases.
During the primal phase, the dual variables remain fixed, and the
algorithm enlarges $M$ as much as possible using only zero-slack edges.
The search is represented by a cherry forest and uses the operations
\texttt{Grow-out}, \texttt{Grow-in}, \texttt{Augment}, and
\texttt{Expand}. Once none of these operations is applicable, the current matching has maximum cardinality on the zero-slack subgraph. After this, we apply \texttt{Shrink} operations and contract every
nontrivial cherry blossom into a supernode. This turns
the cherry forest into an ordinary alternating forest, on which we
perform a dual update. The process
continues until the matching is perfect. This procedure is summarized in Algorithm~\ref{algo_main}. We discuss it more formally below. For a discussion of correctness, see \cref{sec_correctness_proof}.

\begin{algorithm}[!h]
\caption{Blossom VI}\label{algo_main}
    Initialize feasible dual variables and a matching $M$ \\
    Make every unmatched node a root of a cherry tree \\
    \While{$M$ is not perfect}{
        Apply {\tt Grow-out}, {\tt Grow-in}, {\tt Augment}, {\tt Expand} operations while possible\\
        {\tt Shrink} all nontrivial cherry blossoms \\
        Make a dual update
    }
    \Return $M$
\end{algorithm}

\subsection{Initialization}
In the first line of Algorithm~\ref{algo_main}, we initialize the dual variables $y$ and the matching $M$. It is sufficient to take feasible $y$ with $y_S = 0$ for all non-singleton odd-size $S \subset V$ and a matching that uses tight edges. Details of our initialization are given in \cref{section_remarks}.

\subsection{Graph structure}
We maintain a graph $(V,E)$ that is obtained from the original
graph $G_{\tt orig}=(V_{\tt orig},E_{\tt orig})$ by contracting some odd-cardinality subsets of the nodes into supernodes. We store a valid matching $M$ and a cherry forest structure that we discuss below.

\myparagraph{Cherry forest} 
We maintain node sets $V^+\subseteq V$ and $V^-\subseteq V$.
Nodes in $V^+$ are called {\em plus-nodes}, and nodes in $V^-$ are {\em minus-nodes},
and nodes in $V^\pm\eqdef V^+\cap V^-$ are {\em $\pm$-nodes}. The meaning of $V^+$ and $V^-$ is that every plus-node has an even-length alternating path to some unmatched node, and every minus-node has an odd-length alternating path to some unmatched node.
Nodes in $V\setminus (V^+\cup V^-)$ without labels are called {\em unlabeled}.
A node in $V^+\setminus V^-$ is a {\em purely plus-node},
while a node in $V^-\setminus V^+$ is a {\em purely minus-node}.
Unmatched nodes are called {\em roots}; they are purely plus-nodes.

Every non-root plus-node $v$ stores a matched parent edge, called its
\emph{plus-parent}; every minus-node stores an unmatched parent edge,
called its \emph{minus-parent}. We orient each parent edge from the node
towards its parent.
For a plus-node $v$, let $\calP_v^+$ be the maximal path obtained by starting at
$v$ with its plus-parent and then following plus- and minus-parents
alternately. Define $\calP_v^-$ similarly for a minus-node.
% Given this structure, we define a directed graph $(V,A)$,
% where $A$ is the union of all plus-parents and minus-parents, viewed as directed edges.
% A path in this graph is called {\em alternating} if the matched / non-matched status of edges along this path always
% alternates. Equivalently, it follows plus- and minus-parents in an alternating order.
% For a plus-node $v$ let $\calP^+_v$ be the maximal alternating path that starts with the plus-parent of $v$;
% clearly, it is uniquely defined. Similarly, for a minus-node $v$ let $\calP^-_v$ be the maximal alternating
% path that starts with the minus-parent of $v$. 
For a simple path $\calP$ passing through nodes $x,y$, use $\calP[x\rightarrow y]$
to denote the subpath of $\calP$ from $x$ to $y$.

The algorithm will maintain the following key invariants for the cherry forest structure (defined later).
\begin{itemize}[leftmargin=25pt]
\item[\InvA] For each plus-node (resp. minus-node) $v$,  $\calP^+_v$ (resp. $\calP^-_v$) is a simple path ending at a root.
\item[\InvB] Each $\pm$-node $v$ stores a purely plus-node $\r_v\in V$ called the {\em receptacle of $v$}.
This node lies on the paths $\calP^+_v,\calP^-_v$, and 
all nodes on the paths $\calP^+_v[v\rightarrow \r_v]$ and $\calP^-_v[v\rightarrow \r_v]$
excluding $\r_v$ are $\pm$-nodes. 

%\item[\InvB] If $v$ is a $\pm$-node then paths $\calP^+_v$ and $\calP^-_v$ end at the same root node.
%Furthermore, if $(x,y)$ is the first common (directed) arc of $\calP^+_v,\calP^-_v$ then $x$ is a plus-node, and all nodes of paths $\calP^+_v[v\rightarrow x],\calP^-_v[v\rightarrow x]$
%excluding $x$
%are $\pm$-nodes. If $\calP^+_v,\calP^-_v$ have no common arc then we take $x$ to be the endpoint of $\calP^+_v,\calP^-_v$ here.
\end{itemize}

We use the cherry-tree framework of~\cite{cherry_paper}, adapted to
graphs containing contracted supernodes. We restate the definitions
needed in our setting.~\footnote{Our definitions are actually different from those in~\cite{cherry_paper} (but encode the same data structure). We chose to use them because they involve fewer conditions than~\cite{cherry_paper}. Note that our conditions rely on ``global'' properties of paths $\calP^+_v$ and $\calP^-_v$, whereas the conditions in~\cite{cherry_paper} are purely ``local''.}

\begin{definition}[cherry forest, cherry tree, cherry blossom]
    A data structure satisfying \InvA\ and \InvB\ is a
\emph{cherry forest}. The \emph{cherry tree} rooted at $r$ consists
of all labeled nodes whose parent paths end at $r$. For a purely
plus-node $x$, its \emph{cherry blossom} is defined as $\calB_x = \{v\in V^\pm:\r_v=x\} \cup \{x\}$.
\end{definition}
This terminology comes from~\cite{cherry_paper}. In case $\calB_x = \{x\}$, the singleton $\{x\}$ is regarded as a trivial cherry blossom.

\myparagraph{Supernodes} 
% The current graph $G=(V,E)$ is obtained from the original
% graph $G_{\tt orig}=(V_{\tt orig},E_{\tt orig})$ by contracting some subsets of the nodes of odd cardinality to supernodes. 
% Every node $v$ of the current graph represents an odd-size set
% $S \subset V_{\tt orig}$. For an original node $v$, this set is a singleton $\{v\}$.
% A node representing more than one original vertex is
% called a \emph{supernode}. The sets represented by the current
% top-level nodes partition $V_{\tt orig}$.
The nodes of the current contracted graph are called \emph{top-level
nodes}; nodes stored inside a supernode are called \emph{internal
nodes}. Each top-level node $v$ represents an odd-cardinality set
$S(v)\subseteq V_{\tt orig}$, and the sets represented by the
top-level nodes partition $V_{\tt orig}$. If $S(v)$ contains more than
one original node, then $v$ is called a \emph{supernode}.

A supernode stores its internal graph, an internal matching with
exactly one node unmatched, and the corresponding cherry-forest
structure. The unmatched internal node is called the
\emph{receptacle} of the supernode. For an internal node $v$, we treat paths $\calP^+_v$ and $\calP^-_v$ as ending at the receptacle of the corresponding supernode.
% We will say that a node can be either an original node of the input graph or a supernode. 
% Formally, we can define nodes recursively as follows.
% \begin{definition}
% A node $\calB$ is either a node of the original graph, or a set of nodes of odd cardinality
% together with a graph $G_\calB=(\calB,E_\calB)$, a matching $M$ on $G_\calB$ and the associated cherry forest data structure
% satisfying the following: \\
% (i) Graph $G_\calB$ is obtained by taking the subgraph of $G_{\tt orig}$ induced by nodes in $V_{\tt orig}$ corresponding to $\calB$, 
% and then contracting subsets corresponding to nodes $v\in \calB$ to individual nodes. \\
% (ii) $M$ has cardinality $(|\calB|-1)/2$, and hence $\calB$ has a single unmatched node, denoted as $\r_\calB$. It is called the {\em root} or the {\em receptacle} of $\calB$. \\
% (iii) All nodes in $\calB\setminus\{\r_\calB\}$ are $\pm$-nodes.
% \end{definition}

% When we specifically discuss non-elementary nodes, we will call them supernodes. The last invariant guarantees that, inside a given supernode, we can find an alternating path from any node to the receptacle. Thus, nodes inside a given supernode form a cherry blossom.

An important property of a cherry blossom $\calB$ is that its receptacle can be ``moved'' to any other node $v$ of $\calB$.
This operation is called {\em cherry blossom rotation} in~\cite{cherry_paper}. We state the corresponding lemma below and discuss it further in \cref{sec_rotation_proof}.\footnote{We believe that the original proof in \cite{cherry_paper} may have an incomplete termination argument.
In the proof, the authors first show how to do the rotation when paths $\calP^+_v$ and $\calP^-_v$ have no common arcs. 
When a common arc $(x,y)$ exists, they take the first such arc, move the root to $x$ recursively,
and then argue that the length of the path $\calP^+_v$ (which now ends at $x$) strictly decreases, since it is a strict subpath of the original path $\calP^+_v$.
However, recursively applying the rotation operation may have affected plus- and minus-parents on $\calP^+_v[v\rightarrow x]$, and the claim that the new $\calP^+_v$ has become shorter might not hold.}
\begin{lemma}[Receptacle rotation]\label{lemma:rotation}
There exists an algorithm that takes a cherry blossom $\calB$ and a node $v\in\calB\setminus\{\r_\calB\}$
and makes $v$ the root of this cherry blossom (by modifying the matching, labels, and parent pointers). 
%This algorithm runs in $O(|\calB|)$ time.
\end{lemma}

\myparagraph{Slacks and dual variables}
Throughout the algorithm, every matched edge is tight, and for every odd set $S \subset V$ with positive dual variable, its internal matching contains $(|S|-1)/2$ edges. 
Every matched edge and every parent edge, including parent edges stored
inside supernodes, has zero slack. Once $M$ becomes perfect, we have $x(\delta(S)) = 1$ for every $S$ with $y_S > 0$, so the primal and dual solutions satisfy complementary slackness.

For an original node $v$, the dual variable $y_v$ is unrestricted.
For a supernode $v$ representing an odd set $S$, we write $y_v$ for
the blossom variable $y_{S}$ and require $y_v \geq 0$. Primal
updates modify the matching and cherry forest while leaving these
variables unchanged; dual updates modify the dual variables of
top-level nodes. We describe the primal update operations below.

% We denote the current slack of an edge $e$ by ${\tt slack}(e)$; it is given by \cref{eq_slack_def}.
% The following invariant is maintained throughout the algorithm:
% all plus- and minus-parents in the cherry forest and in all supernodes, as well as all matched edges,
% have zero slack. All other edges have a nonnegative slack.

% Each supernode $v$ stores the dual variable $y_v$.
% If $v$ corresponds to a singleton node in $V_{\tt orig}$ then $y_v$ can have any sign.
% Otherwise, if $v$ corresponds to an odd-cardinality subset of $V_{\tt orig}$ with at least three nodes, we must have $y_v\ge 0$.

% The algorithm alternates between (i) primal updates, which modify the current cherry forest, blossoms and matchings while keeping all dual variables intact, and
% (ii) dual updates, which modify variables $y_v$ in the current cherry forest. Below we describe operations that are allowed in these updates.

\subsection{Primal updates}
The primal search uses only tight edges; all edges of positive slack
are ignored. Four operations modify the search structure:
\texttt{Grow-out}, \texttt{Grow-in}, \texttt{Augment}, and
\texttt{Expand}. Exhausting these operations yields a maximum-cardinality matching on the zero-slack subgraph (see \cite{cherry_paper}, Theorem 3.1). Once none of them is applicable, the algorithm uses a
fifth operation, \texttt{Shrink}, to contract the nontrivial cherry
blossoms before the dual update.

\begin{figure*}[htbp]
    % \hspace{-100pt}
    \begin{tabular}{cc}
        \shortstack{\fbox{\includegraphics[scale=0.45]{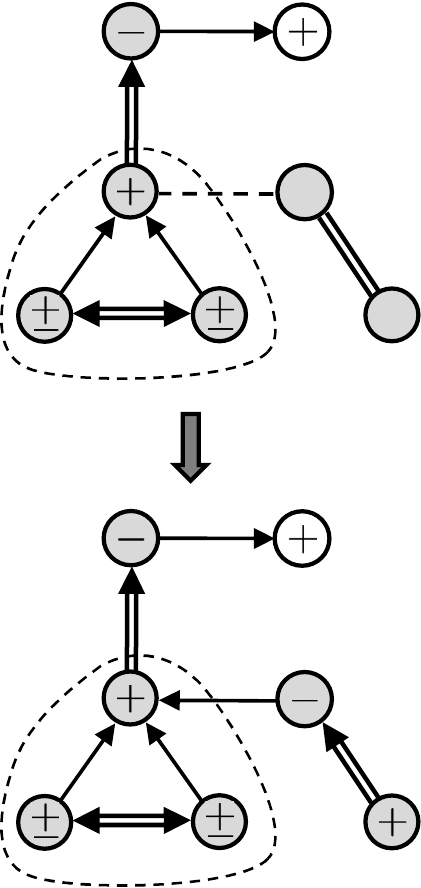} }
        \begin{picture}(0, 0)
        	\put(-68,153){$v$} 
		\put(-29,154){$x$} 
		\put(-12,131){$y$}
	\end{picture}
        \\\tt{(a) Grow-out}} &
        \shortstack{\fbox{\includegraphics[scale=0.45]{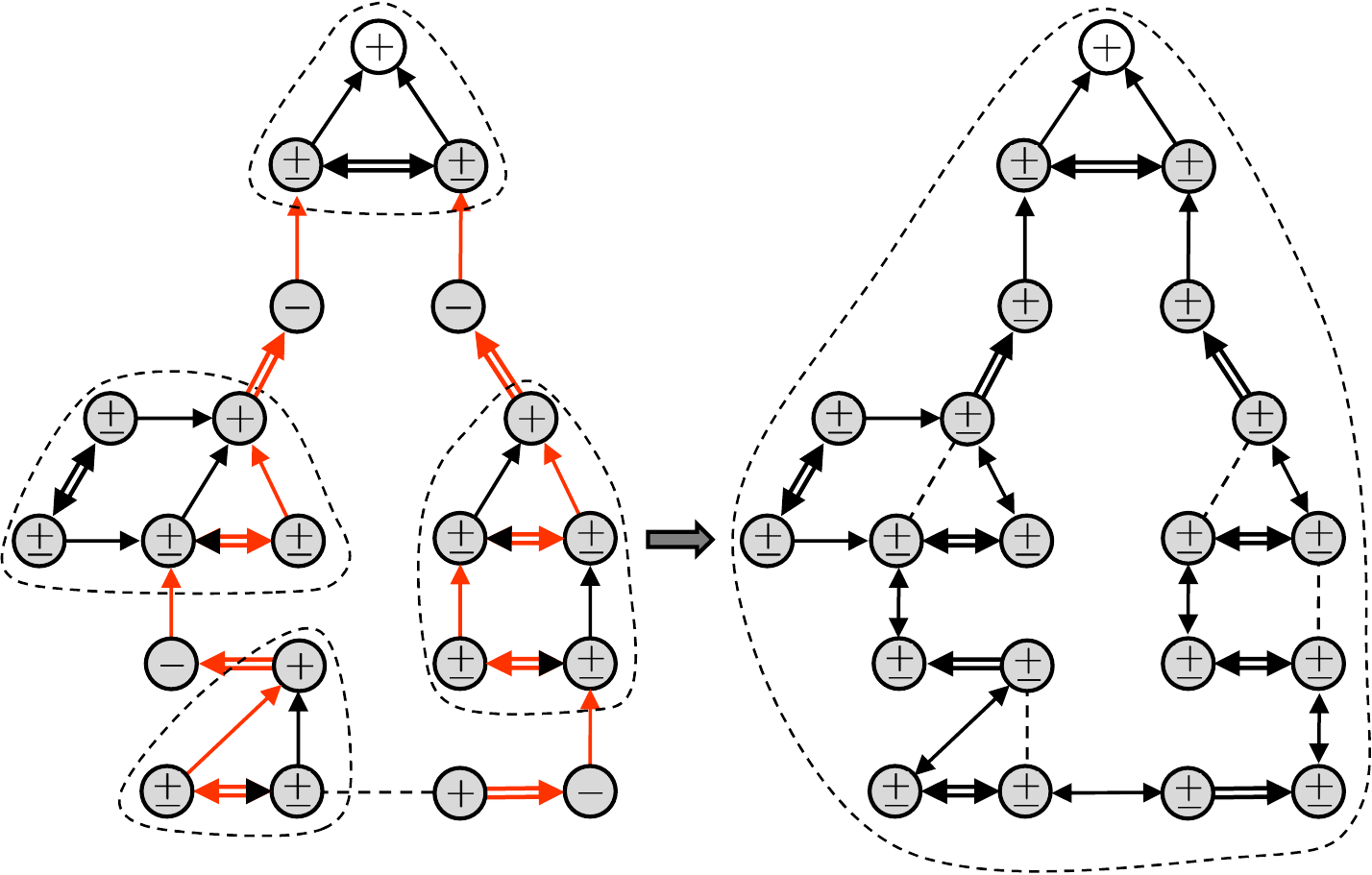}}\\\tt{(b) Grow-in}}
        \begin{picture}(0, 0)
        	\put(-247,39){$u$} 
		\put(-215,40){$v$} 
		\put(-222,207){$a=r$} 
		\put(-272,172){$x_u$} 
		\put(-267,140){$y_u$} 
		\put(-201,172){$x_v$} 
		\put(-205,140){$y_v$} 
        \put(-255,200){$\calB$} 
	\end{picture}
         \vspace{10pt} \\
        {} &
        \hspace{0pt}\shortstack{\fbox{\includegraphics[scale=0.45]{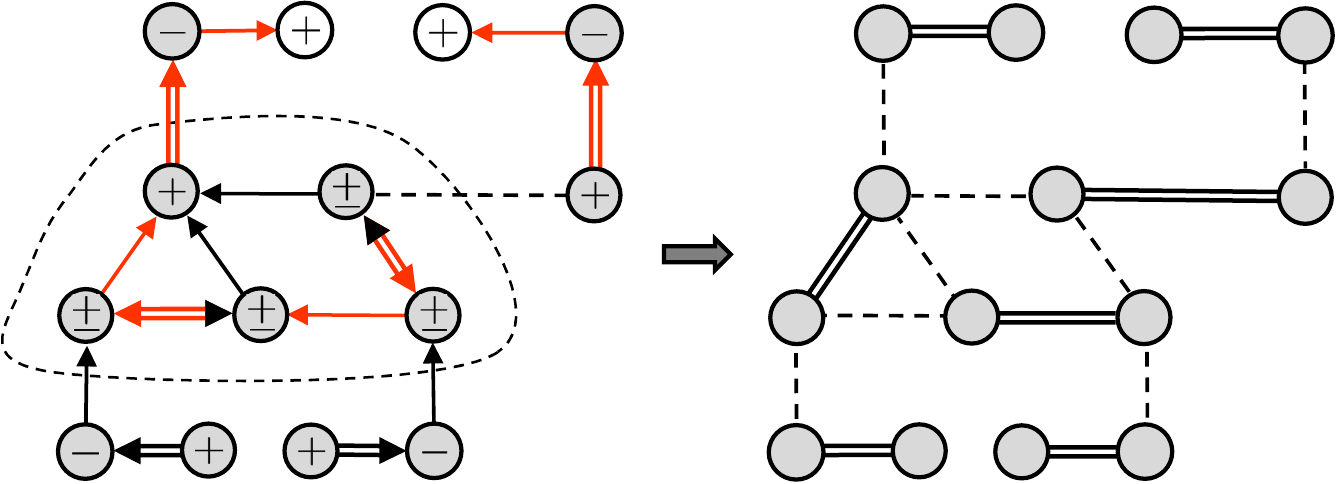}}
        \begin{picture}(0, 0)
        	\put(-215,65){$u$} 
		\put(-179,65){$v$} 
	\end{picture}
        \\\texttt{(c) Augment}} \vspace{-30pt} \\
        \shortstack{\fbox{\includegraphics[scale=0.45]{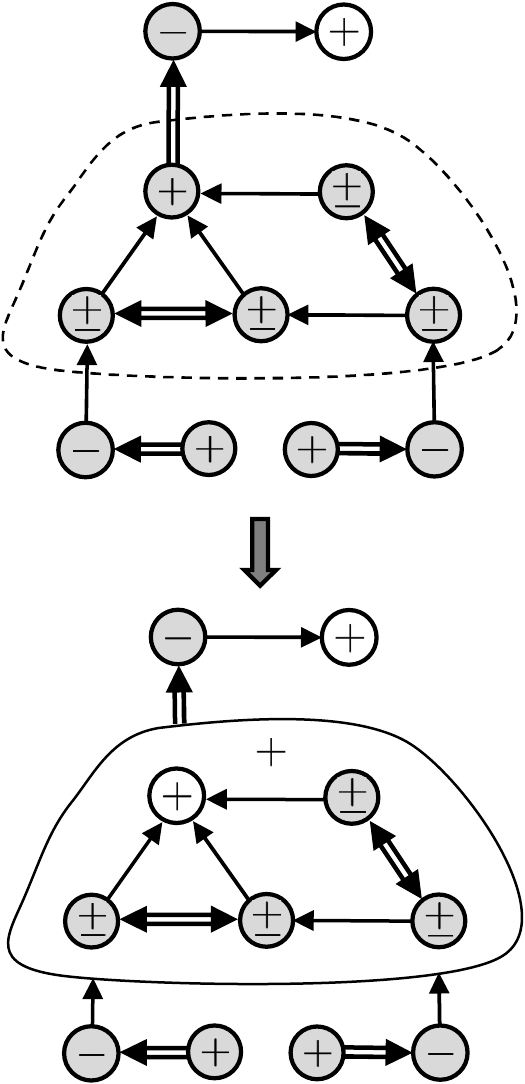}}\\\tt{(d) Shrink}}
        \begin{picture}(0, 0)
        	\put(-93,210){$r$} 
	\end{picture}
         &
        \shortstack{\fbox{\includegraphics[scale=0.45]{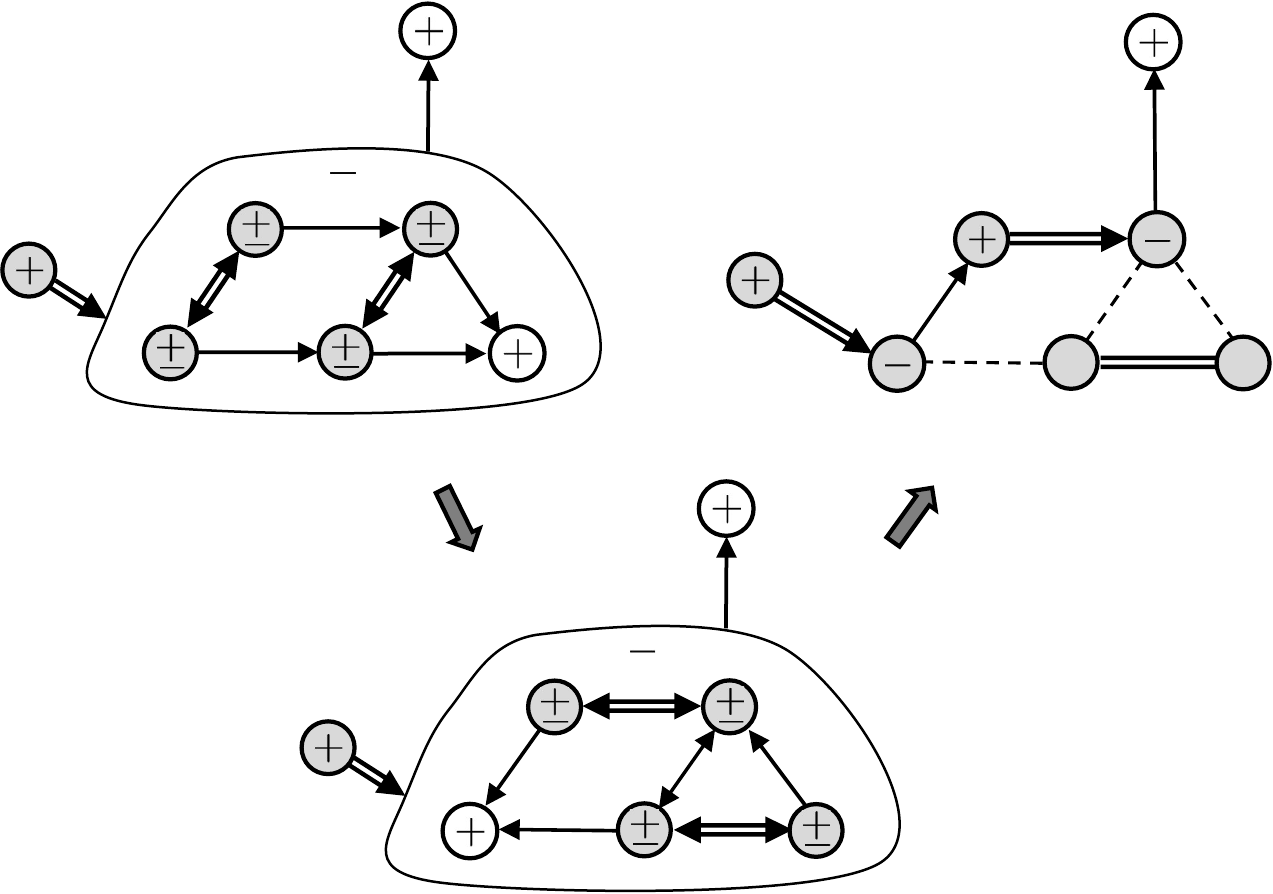}}\\\tt{(e) Expand}} 
        \begin{picture}(0, 0)
        	\put(-223,176){$v$} 
        	\put(-280,156){$u$} 
        	\put(-205,198){$w$} 
        	\put(-251,139){$x$} 
        	\put(-191,167){$y$} 
	\end{picture}
 
    \end{tabular}
    \caption{Primal update operations. White nodes denote roots of cherry trees and receptacles of contracted blossoms, solid nodes are all other nodes.
    Double-line edges are matched, single-line edges are not matched. 
    Edges with arrows are parent edges (which can be bidirectional; e.g., a bidirectional double-line edge $\{x,y\}$ means
    that $(x,y)$ is the plus-parent edge of $x$ and $(y,x)$ is the plus-parent edge of $y$).
    Dashed edges are edges with zero slack that are not parent edges.
    Directed red edges in {\tt Grow-in} and {\tt Augment} denote paths from nodes $u$ and $v$ towards the root
    computed in the respective operation.
    Dashed lines enclose cherry blossoms, solid lines enclose contents of supernodes.
    }    \label{fig:primal}
\end{figure*}

\myparagraph{\underline{\tt Grow-out}} Adds a pair of previously unlabeled nodes to a cherry tree. If $v$ is a plus-node, $\{v, x\}$ is a tight edge, $\{x, y\}$ is a matched edge, and $x,y$ are unlabeled nodes,
then $v$ can ``acquire'' $x$ and $y$ as descendants: $x$ becomes a minus-node with minus-parent $(x,v)$ and $y$ becomes a plus-node with plus-parent $(y,x)$
(see \cref{fig:primal}(a)).

\myparagraph{\underline{\tt Grow-in}} Makes a larger cherry blossom within a given tree. If $u$ and $v$ are plus-nodes in the same cherry tree but in different cherry blossoms and $\{u, v\}$ is a tight edge, then the cherry tree can be grown as follows (see \cref{fig:primal}(b)). \\
(i) Find the first common arc $(a,b)$ of paths $\calP^+_u$ and $\calP^+_v$;
if no such arc exists, take $a$ to be the common endpoint of $\calP^+_u,\calP^+_v$ (it is then the root of the cherry tree).
The edge $ab$ must be matched (if it exists) and $a$ must be a plus-node. \\
(ii) Let $\calB$ be the cherry blossom to which $a$ belongs (possibly a singleton), and let $r=\r_\calB$ be its receptacle.
For each $w\in\{u,v\}$ that satisfies $w\notin\calB$ do the following: find the first arc $(x,y)$ on $\calP^+_w$
such that $x\notin\calB$ and $y\in \calB$. The node $x$ must be a purely minus-node.
All nodes on the path $\calP^+_w[w\rightarrow x]$ become $\pm$-nodes. Set their plus- and minus-parents to point to edges of this path
in the appropriate directions, with two exceptions: keep the minus-parent of $x$ as $(x,y)$,
and set the minus-parent of $w$ to $(w,\bar w)$ where $\bar w\in\{u,v\}\setminus\{w\}$. \\
(iii) Update the receptacle of nodes in all cherry blossoms through which $\calP^+_w[w\rightarrow r]$ passes to be $r$.

\myparagraph{\underline{\tt Augment}} 
Suppose that $u$ and $v$ are plus-nodes in different cherry trees and
that $\{u, v\}$ is tight. Then $({\tt reverse}(\calP^+_u),(u,v),\calP^+_v)$ is an augmenting path between the two roots. Flipping the matching
status of its edges increases $|M|$ by one. All nodes in the two
destroyed cherry trees then become unlabeled (\cref{fig:primal}(c)).

\myparagraph{\underline{\tt Shrink}} Creates a supernode. Suppose that $\calB$ is a cherry blossom with receptacle $r=\r_\calB$. This operation contracts $\calB$ into a single node,
which becomes a purely plus-node with dual variable zero. The node $r$ becomes unmatched inside $\calB$  (\cref{fig:primal}(d)). In our algorithm, \texttt{Shrink} operations happen only when no other primal update is possible.

\myparagraph{\underline{\tt Expand}} Expands a supernode. Suppose that $v$ is a purely minus-node with $y_v=0$ corresponding to a contracted blossom $\calB$.
Let $\{u, v\}$ be the unique matched edge (then $(u,v)$ is the plus-parent of $u$), 
and let $(v,w)$ be the minus-parent of $v$.
Let $\{u, x\}$ and $\{y, w\}$ be the edges with $x,y\in\calB$ corresponding to $\{u, v\}$ and $\{v, w\}$, respectively.
First, apply the rotation operation from \cref{lemma:rotation} to move the root of $\calB$ to $x$.
Then make nodes on the path $\calP^+_y$ purely plus or purely minus nodes in an alternating fashion
(with $x,y$ being purely minus-nodes), set their parents to point to edges along the reverse of this path, and uncontract $\calB$. All other nodes of $\calB$ become
unlabeled. Make $(u,x)$ the plus-parent of $u$, and $(y,w)$ the minus-parent of $y$ 
(\cref{fig:primal}(e)).

\subsection{Dual updates}

%The dual update phase can be applied when there are no top $\pm$-nodes (and hence each cherry tree is a normal alternating tree). In this phase we 
%find a nonnegative number $\delta_T$ for each cherry tree $T$, and then update dual variables for nodes $v$ in $T$: if $v$ is a plus-node then increase $y_v$ by $\delta_T$, and if $v$ is a minus-node then decrease $y_v$ by $\delta_T$.
%This will increase the dual objective by $\sum\limits_T \delta_T$, since in every alternating tree the number of  plus-nodes exceeds the number of minus-nodes by 1. We require the dual variables to remain feasible; this gives constraints of the form 
%$\delta_T\le c_T$ for some trees $T$, and
%$\delta_{T}\pm\delta_{T'}\le c_{TT'}$ for some pairs of trees $T,T'$. A description of these constraints can be found in~\cite{blossom_v}.

After all nontrivial cherry blossoms have been contracted, there are no top-level $\pm$-nodes. Consequently, every cherry tree is an ordinary
alternating tree. We can therefore apply the standard dual
update used in Blossom IV and Blossom V.
We find a nonnegative number $\delta_T$ for each cherry tree $T$, and then update dual variables for nodes $v$ in $T$: if $v$ is a plus-node then increase $y_v$ by $\delta_T$, and if $v$ is a minus-node then decrease $y_v$ by $\delta_T$.
This increases the dual objective by $\sum\limits_T \delta_T$, since in every alternating tree the number of  plus-nodes exceeds the number of minus-nodes by 1. 
We require the dual variables to remain feasible; 
this gives the following constraints on~$\{\delta_T\}$:
\begin{align}\label{eq:deltaT:constraints}
\delta_T+\delta_{T'} & \le \slack(u,v)   && (u,v) \mbox{ is a } (+,+) \mbox{ edge, } \\
                     &                   && \mbox{~~~~~~~~} u \in T, v\in T', T\ne T' \nonumber \\
\delta_T-\delta_{T'} & \le \slack(u,v)   && (u,v) \mbox{ is a } (+,-) \mbox{ edge, } \nonumber \\
                     &                   && \mbox{~~~~~~~~} u \in T, v\in T', T\ne T' \nonumber \\
\delta_T              \le &\;\slack(u,v)   && (u,v) \mbox{ is a } (+,\varnothing) \mbox{ edge, } u \in T \nonumber  \\
\delta_T              \le &\;\slack(u,v)/2 && (u,v) \mbox{ is a } (+,+) \mbox{ edge, } u,v \in T \nonumber  \\
\delta_T              \le &\;y_v           && v\mbox{ is a minus-supernode,  } v \in T  \nonumber 
\end{align}
where we shortened the words ``plus'', ``minus'' and ``unlabeled'' to symbols ``$+$'',
``$-$'' and ``$\varnothing$'', respectively.

\section{Comparison with prior implementations}\label{sec_prior_implementations}
%In this section we give a brief overview of the Blossom V algorithm, which is our main competitor.
%It follows the standard Edmonds framework for minimum cost perfect matchings. 
Existing implementations follow 
the standard Edmonds framework for minimum weight perfect matchings.
In this framework
there are no $\pm$-nodes at the top level;
more precisely, when an odd cycle is discovered by an operation analogous to \texttt{Grow-in}, it is immediately contracted into a supernode called a blossom,
%when such nodes are created by the {\tt Grow-in} operation, the obtained cherry blossom (which is an odd cycle) is immediately contracted to a supernode (called a ``blossom''),
which becomes a purely plus-node.
Thus, the supernodes have a much simpler structure:
they are just odd cycles.

Below we give a short overview of three representative implementations of this framework.

\myparagraph{Blossom IV (Cook-Rohe~\cite{blossom_4})}
This work introduced the ``variable $\delta$'' approach that allows different dual updates $\delta_T$ for different trees. 
More specifically, Cook-Rohe compute connected components in the undirected graph over trees
where two trees $T,T'$ are connected by an edge
if there is a constraint $\delta_T-\delta_{T'}\le 0$ between them.
They then go through these components $\calC$ in some order and greedily increase the values $\delta_T$ for all $T\in \calC$  by the same amount subject to the constraints~\eqref{eq:deltaT:constraints}.
These constraints %in~\eqref{eq:deltaT:constraints} 
are computed by explicitly traversing all edges.

\myparagraph{Mehlhorn-Sch{\"a}fer~(\cite{mehlhorn_schafer})}
This work implemented the algorithm
of Galil-Micali-Gabow~\cite{mwpm_theory_gabow3} whose complexity is $O(mn\log n)$ (or $O(m\log n)$ per augmentation). 
It avoids an explicit traversal of edges in dual updates by storing edges in a {\em concatenable priority queue}.
%This speeds up the computation of the constraints in~\eqref{eq:deltaT:constraints}
%since the edges no longer need to be traversed explicitly. 
The algorithm uses a ``fixed $\delta$'' approach, i.e., sets $\delta_T$ to the same value for all trees. 

\myparagraph{Blossom V (Kolmogorov~\cite{blossom_v})}
This work combined the variable $\delta$ approach with the use of priority queues.
The algorithm maintains an auxiliary graph $\calG$ whose vertices are the trees, where two trees $T,T'$ are connected if there is an edge $\{v,v'\}$ with $v\in T$, $v'\in T'$. Each edge $\{T,T'\}$ of $\calG$ stores priority queues $pq^{++}(T,T')$, $pq^{+-}(T,T')$, $pq^{-+}(T,T')$
containing  edges between $T$ and $T'$ at the top level in $G$ of the appropriate type. Similarly, nodes $T$ of $\calG$ store
priority queues $pq^{++}(T)$, $pq^{+\varnothing}(T)$ of edges incident to $T$, and a priority queue $pq^{-}(T)$ of blossoms in $T$.

Our implementation uses the variable $\delta$ approach and maintains an auxiliary graph over trees with priority queues, as in Blossom V.
Also, we use the same strategy for computing dual updates (based on connected components) as in Blossom IV and Blossom V. As stated before, the main novelty of our implementation is in how the blossoms are handled in the primal phase.

\section{Experimental results}

We have measured the performance of our C++ implementation against the performance of Blossom V on 9 families of instances. We used a laptop running Ubuntu 24.04 with 16 GB of RAM and an Intel Core i7-8565U CPU (1.80 GHz, 8 MB cache). The code was compiled using GCC 13.3.0 with the -O3 optimization flag. Both solvers are single-threaded. We used Blossom V version 2.05 available at \url{https://pub.ista.ac.at/~vnk/software.html}. The results are summarized in \cref{fig_runtimes}, \cref{table_tsp}, and \cref{fig_runtime_ratios}. Every instance has been solved 3 times by both algorithms, and we report the average runtime. We never observed a deviation higher than $5\%$ in the runtimes for any given instance. We used the following instance families.

\begin{enumerate}
    \item \textbf{Dense random} instances. Erd\H{o}s--R\'enyi random graphs with average node degree $n/10$. To ensure the existence of a perfect matching, we sample a random perfect matching first, and then sample the rest of the edges uniformly. The edge weights are sampled randomly from $[-10^6, 10^6]$.

    \item \textbf{Sparse random} instances. Erd\H{o}s--R\'enyi random graphs with average node degree $10$. Similarly to the dense random instances, we sample a random perfect matching first, and then sample the rest of the edges uniformly. The edge weights are sampled randomly from $[-10^6, 10^6]$.

    \item \textbf{Delaunay-big-weights} instances. We sample $n$ points from a square of size $10^6 \times 10^6$ and compute the Delaunay triangulation of these points. The points and the triangulation edges define the instance. Edge weights are lengths of the edges rounded to the nearest integer. It was proved in \cite{delaunay_contains_matching} that Delaunay triangulations of even-sized general-position point sets contain perfect matchings.
    
    \item \textbf{Delaunay-small-weights} instances. Same as Delaunay-big-weights, but the $n$ points are sampled from a $\sqrt{n} \times \sqrt{n}$ square. This results in most of the edge weights being small integers.

    \item \textbf{Geometric-big-weights} instances. We sample $n$ points from an $a \times a$ square, $a = 10^6$, and draw edges between points with pairwise distances at most $a \sqrt{\frac{10}{\pi n}}$. This choice of the threshold results in an average degree of approximately 10. To ensure that a perfect matching exists, we take the union of the graph and a random perfect matching. The edge weights are edge lengths rounded to integers.
    
    \item \textbf{Geometric-small-weights} instances. Same as Geometric-big-weights, but the $n$ points are sampled from a $\sqrt{n} \times \sqrt{n}$ square. The average degree is still around 10, but the edge weights are mostly small integers.

    \item \textbf{Maxcut-big-weights} instances. We sample $n$ points from a sphere and compute their convex hull. This defines a planar triangulation (with the outer face also triangular). We set the edge weights of the triangulation to be random integers from $0$ to $10^6$. One can compute the max-cut of planar graphs by solving the minimum weight perfect matching problem on some auxiliary graph \cite{hadlock_planar_maxcut}. In the case of triangulations, we can build the auxiliary graph as follows. Take the dual planar graph with weights inherited from the weights of the triangulation edges. Replace each node with a triplet of nodes (so that the new graph is still 3-regular). Within each triplet, the nodes are connected with zero-weight edges. A minimum weight perfect matching in this graph corresponds to a max-cut in the triangulation. Solving max-cut is also sometimes described in terms of finding a ground state of a specific Ising model.

    \item \textbf{Maxcut-small-weights} instances. Same as Maxcut-big-weights, but the weights of the triangulation edges are sampled randomly from $\{0, 1\}$.

    \item \textbf{TSP} instances. Following the Blossom V paper, we used the set of TSP instances by Andre Rohe obtained from \url{https://www.math.uwaterloo.ca/tsp/vlsi/index.html}. We only considered the instances with more than 100,000 nodes. If an instance contained an odd number of points, we removed the last point. For each instance, we computed the Delaunay triangulation of the point cloud and set the weights to edge lengths rounded to integers.
\end{enumerate}

\begin{figure*}[!h]
    \centering
    \includegraphics[width=1.01\linewidth]{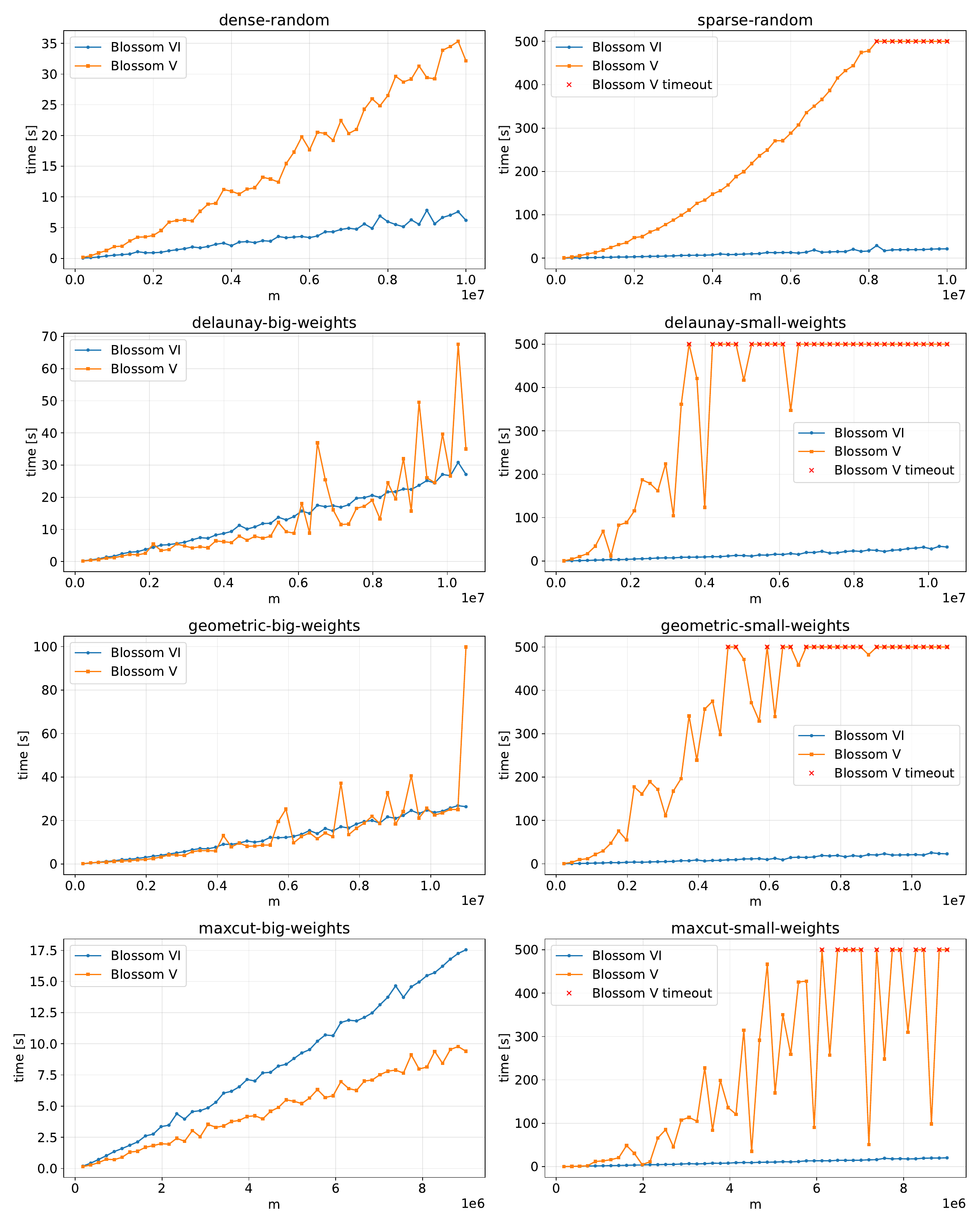}
    \caption{Runtimes of Blossom VI and Blossom V. Red crosses denote 500-second timeouts.}
    \label{fig_runtimes}
\end{figure*}

In our experiments, we set a 500-second timeout. Our code outperforms Blossom V on dense random instances, sparse random instances, Delaunay-small-weights, geometric-small-weights, and maxcut-small-weights instances. It performs approximately the same on Delaunay-big-weights and geometric-big-weights instances, and it is about two times slower on maxcut-big-weights instances. 

One may observe that on dense random, Delaunay-big-weights, geometric-big-weights, and maxcut-big-weights families, both algorithms scale roughly linearly over the tested range, taking tens of seconds to terminate for instances with $\sim 10^7$ edges. On sparse random, Delaunay-small-weights, geometric-small-weights, and maxcut-small-weights families, Blossom V hits timeouts on larger instances, while Blossom VI still takes at most tens of seconds. 

\begin{table}[h]
\caption{Runtime comparison on TSP instances. The runtimes are given in seconds.}
\label{table_tsp}
\centering
\small
\begin{tabular}{|c|c|c||c|c|}
\hline
name & $n$ & $m$ & B5 time& B6 time\\
\hline
% dan59296 & 59296 & 177299 & 0.07 & 0.09 \\
sra104815 & 104814 & 314222 & 0.20 & 0.23\\
ara238025 & 238024 & 713594 & 8.19 & 0.83  \\
lra498378 & 498378 & 1494967 & 8.57 & 1.99 \\
lrb744710 & 744710 & 2233725 & 120.27 & 2.45 \\
\hline
\end{tabular}
\normalsize
\end{table}

To explain runtime differences, we collected additional statistics for some runs in
\cref{tab_stats}. We include at least one instance from each family. For families in which Blossom V showed high runtime dispersion, we include both a faster and a slower
Blossom V run. All selected runs were completed within the time limit.

For each run, we report the runtime, the number of supernodes created,
and the maximum supernode nesting depth observed during the run. For
Blossom V, we additionally report the fractions of its runtime spent
on initialization and on \texttt{Expand} and \texttt{Shrink}
operations.\footnote{A similar operation breakdown is not meaningful for
Blossom VI. In Blossom V, a primal operation immediately performs the
associated priority-queue and slack updates. Blossom VI performs these
updates at the end of the primal phase, and they cannot be
linked unambiguously to an individual primal operation.}

\begin{table*}[!h]
    \caption{Statistics comparison between Blossom V and Blossom VI. For every instance, $n$ and $m$ are given in millions; ``-bw'' and ``-sw'' denote ``-big-weights'' and ``-small-weights'', respectively. B5 stands for Blossom V, B6 stands for Blossom VI.
}
    \label{tab_stats}
    \centering
    % \small
    \setlength{\tabcolsep}{1.pt}
    \renewcommand{\arraystretch}{1.}
    
    \begin{tabular}{!{\vrule width 1.pt}c|c|c!{\vrule width 1.pt}c|c!{\vrule width 1.pt}c|c!{\vrule width 1.pt}p{2cm}|c!{\vrule width 1.pt}c|c|c!{\vrule width 1.pt}}
    \Xhline{1.pt}
    \multicolumn{3}{!{\vrule width 1.pt}c!{\vrule width 1.pt}}{instance}
    & \multicolumn{2}{c!{\vrule width 1.pt}}{runtime [s]} 
    & \multicolumn{2}{c!{\vrule width 1.pt}}{\# supernodes} 
    & \multicolumn{2}{c!{\vrule width 1.pt}}{max supernode depth}
    & \multicolumn{3}{c!{\vrule width 1.pt}}{B5 operation times} \\
    \hline
    family & $n$ [$10^6$] & $m$ [$10^6$] & B5 & B6 & B5 & B6 & \centering B5 & B6 & init & expand & shrink \\
    \Xhline{1.pt}
    random-dense & 0.0141 & 10.0 & 34.3 & 6.31 & 0 & 2 & \centering 0 & 1 & 100\% & 0\% & 0\% \\
    random-sparse & 1.00 & 5.00 & 231 & 9.82 & 7 & 574 & \centering 4 & 566  & 100\% & 0\% & 0\% \\
    delaunay-bw & 3.36 & 10.1 & 30.9 & 27.0 & 719,024 & 691,603 & \centering 3,454 & 444 & 10\% & 34\% & 5\% \\
    delaunay-bw & 3.43 & 10.3 & 93.5 & 31.9 & 783,749 & 727,562 & \centering 12,867 & 607 & 3\% & 53\% & 3\% \\
    delaunay-sw & 0.980 & 2.94 & 229 & 7.48 & 228,699 & 66,105 & \centering 1,270 & 13 & 0\% & 96\% & 0\% \\
    geometric-bw & 1.96 & 10.8 & 29.7 & 26.1 & 521,915 & 499,881 & \centering 3,686 & 444  & 10\% & 16\% & 6\% \\
    geometric-bw & 2.00 & 11.0 & 111 & 26.4 & 531,093 & 506,195 & \centering 5,111 & 452 & 3\% & 76\% & 2\% \\
    geometric-sw & 0.760 & 4.17 & 369 & 6.6 & 213,945 & 48,873 & \centering 1,487 & 7 & 0\% & 95\% & 0\% \\
    maxcut-bw & 6.00 & 9.00 & 10.2 & 18.0 & 2,147,749 & 2,161,210 & \centering 711 & 720 & 10\% & 6\% & 12\% \\
    maxcut-sw & 2.88 & 4.32 & 329 & 9.37 & 724,571 & 539,529 & \centering 2,444 & 21 & 0\% & 97\% & 0\% \\
    maxcut-sw & 3.00 & 4.50 & 42.2 & 9.13 & 727,164 & 555,294 & \centering 3417 & 25 & 1\% & 77\% & 1\% \\
    \Xhline{1.pt}
    \end{tabular}
    \normalsize
\end{table*}

\myparagraph{Expansion-dominated instances}
Except for two random graph instances, Blossom V is slower on those instances that are dominated by performing {\tt Expand} operations. Expanding a supernode is an expensive operation in both the Blossom V and Blossom VI implementations, since it involves the invalidation of the head and tail pointers for all the edges incident to the supernode, including loops. It could be especially expensive for deep supernodes with a lot of intermediate zero-variable supernodes, that get a cascade of expansions within a single primal phase. This is also a potential explanation for high variance in the Blossom V runtimes for some families. We observe that on those instances that are hard for Blossom V, the maximum node depth of Blossom VI is much smaller, generally by one to two orders of magnitude. We believe that this is the main reason why our code outperforms Blossom V on the hard instances: shallow blossoms prevent us from the heavy loss of time on {\tt Expand} operations.

\myparagraph{Initialization-dominated instances} For random instances, Blossom V spends almost all runtime
in the initialization. Its fractional initialization solves the random instances almost to optimality (in the sparse instance, leaving two unmatched nodes after init, and in the dense instance, solving the problem to optimality). Blossom VI uses a less thorough fractional initialization strategy, which runs faster; however, more work is done in the main solving procedure. Thus, our speedup on the random families is explained by the initialization strategy.

\myparagraph{Number of supernodes}
On the big-weight families, the two implementations create similar
numbers of supernodes. On the small-weight families, Blossom VI creates
fewer supernodes. One possible explanation is that
small integer weights produce more edges with the same slack and hence
a denser $E_0$ subgraph. Thus, we have larger cherry blossoms and fewer {\tt Shrink} operations.

\section{Conclusions}
\label{sec:conclusions}

We presented Blossom VI, a practical algorithm for
minimum weight perfect matching. Our main idea is to treat
each primal phase as an unweighted maximum-matching computation on the
tight-edge subgraph. Blossom VI uses
cherry trees and postpones contraction until the primal phase is
complete. It then contracts entire cherry blossoms rather than sequences of nested traditional blossoms.

Our experiments show that this approach is effective on
instances on which Blossom V shows superlinear runtime. On such
instances, Blossom VI creates substantially shallower supernodes and avoids costly cascades of expansions observed in
Blossom V. On families for which Blossom V already scales nearly
linearly, Blossom VI remains competitive, although it can be slower.

\begin{figure*}[!h]
    \centering
    \includegraphics[width=1.\linewidth]{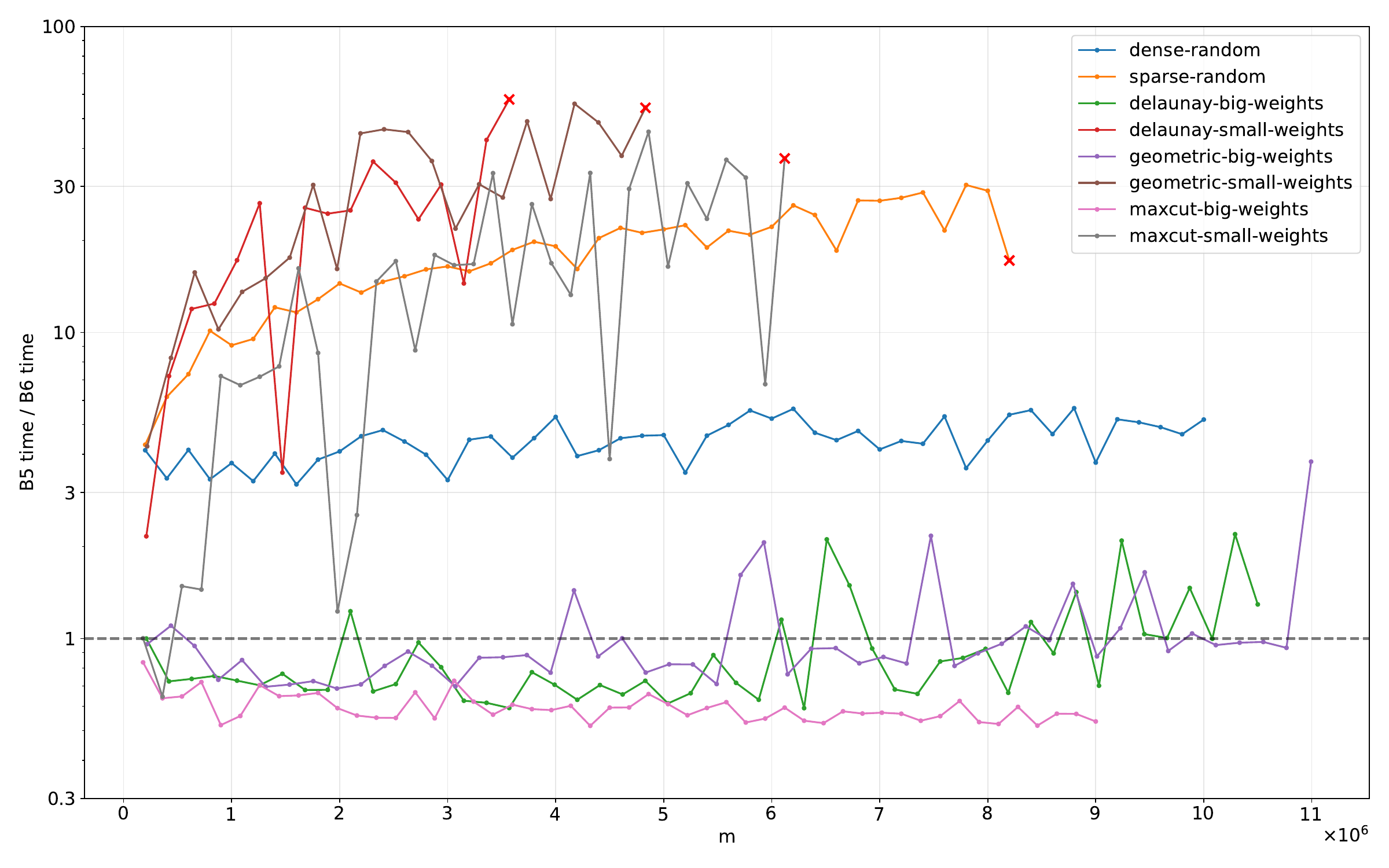}
    \caption{Ratio of Blossom V runtime to Blossom VI runtime (B5$/$B6). For families where Blossom V reached the timeout, the trend is cut at the first timeout, marked with a red cross.}
    \label{fig_runtime_ratios}
\end{figure*}

\newpage
\appendix

\section{Discussion of correctness}\label{sec_correctness_proof}
We use a known correctness result: the
usual multi-tree primal-dual blossom algorithm with feasible dual updates is correct~\cite{blossom_4,blossom_v}.
The new point that must be justified is that contracting an entire
cherry blossom at the end of a primal phase can be reduced to the
traditional blossom framework. First, we show invariants maintained by Blossom VI.

\begin{lemma}
\label{lemma:primal_dual_invariants}
Throughout the algorithm:
\begin{enumerate}[label=(\roman*)]
    \item the dual solution is feasible;
    \item $M$ is a matching and every edge of $M$ is tight;
    \item every parent edge in every top-level or internal cherry
          forest is tight;
    \item every odd set with positive dual variable is represented by
          a possibly nested supernode whose internal matching leaves
          exactly one node unmatched.
\end{enumerate}
\end{lemma}

\begin{proof}
The invariants hold after initialization. The operations
\texttt{Grow-out} and \texttt{Grow-in} change only labels and parent
pointers and use only tight edges. The operation \texttt{Augment}
flips the matching status along a tight alternating path, and hence
preserves both the matching property and tightness of matched edges.

A newly created supernode has dual variable zero, so
\texttt{Shrink} does not violate dual feasibility or condition~(iv).
The operation \texttt{Expand} is applied only to a supernode whose
dual variable is zero. Receptacle rotation changes the internal
matching but continues to leave exactly one internal node unmatched.
Finally, the constraints in \cref{eq:deltaT:constraints} ensure that
a dual update preserves nonnegative edge slacks and nonnegative
supernode dual variables. Parent edges remain tight because their
endpoints receive cancelling dual changes.
\end{proof}

Thus, if Blossom VI returns a matching $M$, then $M$ is indeed a minimum weight perfect matching. This follows from primal feasibility, dual feasibility, and complementary slackness upon termination. It remains to show that Blossom VI terminates and returns $M$ assuming that a perfect matching exists. We will show that the contractions performed by
Blossom VI can be reproduced using traditional blossom contractions.

\begin{definition}[Traditional realization]
A \emph{traditional realization} of a cherry blossom $\calB$ is a
sequence of standard blossom contractions, each applied to an odd
alternating cycle of tight edges, that contracts precisely the nodes
of $\calB$ into one supernode and leaves the matching outside $\calB$
unchanged.
\end{definition}

\begin{lemma}
\label{lemma:traditional_realization}
Every cherry blossom admits a traditional realization.
\end{lemma}

\begin{proof}
We use induction on the number of top-level nodes in the cherry
blossom. The base case is a trivial cherry blossom, for which the claim holds.

Let $\calB$ be nontrivial, with receptacle $r$. Choose any node $v \in \calB \setminus \{r\}$ such that $(v, r)$ is the minus-parent of $v$. Then $\{v, r\}$ is unmatched and tight. $v$ is also a plus-node. By \InvB, $\calP_v^+[v \to r]$ is an even-length alternating path from $v$ to $r$. All edges of $\calP_v^+[v \to r]$
are tight. Therefore $C := \calP_v^+[v \to r] \cup \{\{v, r\}\}$
is an odd alternating cycle of tight edges and may be
contracted as a traditional blossom.

Let $b$ denote the node obtained by contracting $C$. After replacing
$C$ by $b$ and inheriting all parent pointers outside $C$, the set $\calB'=(\calB\setminus V(C))\cup\{b\}$
is again a cherry blossom, now with receptacle $b$. Moreover, $\calB'$ has fewer
top-level nodes than $\calB$. The induction hypothesis gives a
traditional realization of $\calB'$. With the contraction of
$C$, this gives a traditional realization of $\calB$.
\end{proof}

\begin{theorem}[Correctness of Blossom VI]
\label{theorem:correctness}
If Blossom VI reaches a state with $M$ non-perfect and in which no
primal operation and no progress-making feasible dual update is
available, then the input graph has no perfect matching. If an unbounded dual update is discovered, then the input graph also has no perfect matching.
\end{theorem}

\begin{proof}
Consider an imperfect state after the primal phase. By
\cref{lemma:traditional_realization}, replace every contracted cherry
blossom by a nested family of traditional blossoms. This produces a
valid state of a traditional multi-tree blossom algorithm with the
same top-level contracted graph, matching, alternating forest, edge
slacks, and dual variables.

Blossom VI cannot apply any primal-phase operations, so the corresponding traditional state admits no grow,
augment, shrink, or expand operation. The feasible
dual update is determined by the same constraints
\cref{eq:deltaT:constraints}. Therefore, if Blossom VI has no
progress-making feasible dual update, neither does the traditional
algorithm. For a traditional primal-dual blossom
algorithm, a state with non-perfect $M$ with neither a primal nor a dual update possible implies that no perfect matching exists.

Unbounded dual problem implies infeasible primal problem, so if an unbounded dual update is possible, then the graph has no perfect matching.
\end{proof}

Thus, if at any point we do not discover a non-trivial dual update, or the dual update is unbounded, then we can report that no perfect matching exists in the input graph.

We can argue polynomial-time termination in the same way. The traditional multi-tree blossom algorithm, with the same dual phase procedure, terminates after polynomially many operations whenever a perfect matching exists~\cite{blossom_4,blossom_v}. Then a superpolynomial-time execution
of Blossom VI can be translated into a superpolynomial execution of the traditional
algorithm, which is impossible.

% Our algorithm deviates significantly from the usual blossom algorithms with traditional alternating trees and traditional blossoms. For this reason, we need to discuss the correctness of the algorithm.

% It is enough to prove that if no progress can be made in both the primal and the dual phases (and thus, the algorithm is stuck), then the input graph does not contain a perfect matching.

% A convenient proof strategy is to show that the behavior of our algorithm can be emulated by a traditional multi-tree blossom algorithm. Thus, it is enough to show that any cherry blossom can be represented by a traditional blossom (possibly, with sub-blossoms). Then, the state of the trees after the shrinking phase can be realized as a state of a traditional blossom algorithm. Thus, the traditional blossom algorithm would also be stuck in this state.

% Finally, it is left to show that for any cherry blossom $\calB$, there exists a sequence of shrinks of traditional blossoms that also result in shrinking $\calB$. Let $r$ be the receptacle of $\calB$. Take any node $v \in \calB - r$ such that $r$ is a minus-parent of $v$. Let $\calP$ be the even alternating path from $v$ to $r$. $\calP$ together with the minus-parent of $v$ make an odd alternating cycle of zero-slack edges. Shrink it into a traditional blossom and make the blossom the new receptacle of $\calB$. Clearly, the cherry blossom structure of the new smaller cherry blossom is still valid. Repeat this procedure until the entire cherry blossom is contracted. Thus, the correctness claim holds.

\section{Implementation remarks}
\label{section_remarks}

Our implementation includes some optimizations not discussed in the main algorithm description in \cref{sec_algo_overview}. Here, we give a list of important implementation details that were not stated before.

\begin{enumerate}
    \item \textbf{Initialization:} we tested several initialization strategies. 

    First, we considered the greedy initialization. In the very beginning, we initialize the dual variables such that all the edge slacks are nonnegative. We do this by setting $y_v = \frac{1}{2} \min\limits_{e \in \delta(v)} w_e$. Notice that there are no isolated vertices, otherwise, we can immediately report that no perfect matching exists. After that, we go through the nodes and greedily increase the associated dual variables as far as we can (i.e., until some incident edge hits zero slack). The matching is also initialized greedily: if the current node is unmatched and there is an unmatched neighbor at the end of a zero-slack incident edge, we include such an edge in the matching.

    Second, we considered the fractional matching initialization. We still start by initializing $y_v$ such that all slacks are nonnegative. After that, we solve the minimum weight fractional perfect matching problem, obtaining a half-integral solution, and then round it. This initialization was suggested in \cite{fractional_init_suggest}. We have implemented the algorithm described in \cite{fractional_init_description}, with a single alternating tree.

    Third, we considered an approximate fractional matching initialization. For the sake of initialization, we do not have to solve the fractional problem to optimality. Thus, we introduced a threshold $t$: if the alternating tree from the fractional matching algorithm grows to have more than $t$ nodes, we abandon it without performing any augmentations. We found $t = 100$ to be a good value that balances a cheap initialization and a good jumpstart for the main algorithm. We use this initialization strategy in all tests. Enabling the threshold is especially impactful in the sparse random family of instances. Potentially, implementing a fractional matching initialization with multiple trees could give another mild speedup, but we did not implement it.

    \item \textbf{Dual updates:} In the dual update phase, we need to find a collection of values $\delta_T$ to change the dual variables. The dual solution must remain feasible. In general, $\delta_i$ must satisfy the following constraints.
    \begin{equation}
        \begin{cases}
            \delta_i \geq 0 \\
            \delta_i \leq a_i \\
            \delta_i + \delta_j \leq b_{ij} \\
            \delta_i - \delta_j \leq c_{ij} \\
        \end{cases}
        \label{eq_delta_constraints}
    \end{equation}
    The coefficients $a_i$, $b_{ij}$, $c_{ij}$ come from edge slacks and minus-supernode variables. We have tested two dual update approaches, both of which were already present in Blossom V. For a more detailed description, we refer the reader to \cite{blossom_v}.

    The first approach is to compute the connected components in the graph with the edge set $\{(i, j) \; | \; c_{ij} = 0 \}$, and define $\delta_i$ to be the same for all trees in a given connected component. Once the components are computed, we go through the components and greedily increase $\delta_i$ while it is possible.

    The second approach is to solve the LP with constraints defined by \cref{eq_delta_constraints} and the objective of maximizing the dual progress, i.e., $\sum_T \delta_T$. This LP can be converted to an instance of a min-cost flow if one introduces the new variables $\delta_i^+$, $\delta_i^-$, satisfying $\delta_i = \frac{\delta_i^+ - \delta_i^-}{2}$. The optimal solution for $\delta_i$ is half-integral if the coefficients $a_i$, $b_{ij}$, $c_{ij}$ are integers. We have implemented a straightforward min-cost flow algorithm using successive shortest paths.

    The LP approach makes better dual progress in general, but is also more expensive. We tested the connected components approach against the LP approach. Since solving the LP exactly for a large number of trees is expensive, we only used it when the number of trees became less than some threshold. We experimented with thresholds of order $10^2$. Enabling the LP solutions did not affect the runtime. So, in our final version of the code, we use the connected components approach, because it is simpler.

    Another remark regarding the dual updates is the (quarter-)integrality of our dual variables. In our code, we deal with integer weights. It has long been known that there exists an integral optimal dual solution for any minimum weight perfect matching problem with integer edge weights \cite{cunningham_integral_dual}. A single-tree implementation from \cite{fractional_init_description} uses dual variables that remain half-integral throughout the algorithm. One can observe that in our case, if the current edge slacks are even and the dual variables are integer, then the coefficients $a_i$, $b_{ij}$, $c_{ij}$ are integers. 
    Furthermore, making a dual update with suitable integer $\delta_T$ values would yield even edge slacks again. For this reason, when we use the LP approach, we round the half-integer values $\delta_T$ down to integers (and fall back to the connected-components approach in case no progress was made because of the rounding). Consequently, if we quadruple all edge weights in the very beginning, then we will have integer slacks and variables throughout the algorithm.\footnote{We have quarter-integrality instead of half-integrality because of the first step of the initialization, where we set $y_v = \frac{1}{2} \min\limits_{e \in \delta(v)} w_e$. This adds an extra halving of the integer weights.} 

    \item \textbf{Lazy updates for slacks and dual variables:} in the dual update phase, for a tree $T$, we need to increase the dual variables of the top plus-nodes by $\delta_T$ and decrease the dual variables of the top minus-nodes by $\delta_T$. This operation also changes the slacks of the edges incident to the tree. Just like in Blossom V, we do this lazily in $O(1)$ time by storing a variable $y_T$ in the tree holding the sum of $\delta_T$ over all dual updates. A node remembers a lazy variable $\overline{y_v}$, and the true dual variable can be calculated as $y_v = \overline{y_v} \pm y_T$, where $T$ is the tree to which $v$ belongs, and the sign is determined by the sign of $v$. A similar lazy calculation applies to edge slacks. 
    
    In the edge heaps, the edges are compared by their lazy slacks, and nodes in the node heaps are compared by their lazy variables. This is correct, because within the same type of edges or nodes, lazy slacks or variables yield the same ordering as the true slacks or variables.

    \item \textbf{Lazy updates for heads and tails:} every edge remembers the top nodes of its endpoints. Whenever we shrink a set of nodes $S$ into a supernode, the endpoints of the edges in $\delta(S)$ change. We do not update the endpoints immediately. Rather, we allow the head/tail fields of an edge to point to non-top blossoms. Once needed, we can find the true head/tail by walking up the blossom parent pointers. The same is done in Blossom V.

    \item \textbf{Heap choice:} We compared two types of heaps for the edge heaps and the node heaps: the traditional $d$-ary array-based heap and the pairing heap. 

    The traditional array-based heap allows one to push new elements to the heap and pop the heap root in $O(\log N)$ time, where $N$ is the size of the heap. Accessing the root is $O(1)$. Such heaps are attractive, because the elements are stored contiguously in a vector, making the heap more cache-friendly than tree-based heaps. For our purposes, we need an additional operation: given an edge, remove it from the heap that it is currently in. For this, we used handles. Insertion returns a pointer to a handle of the inserted element, and one can use this handle later to locate the element in the heap and remove it. 

    The pairing heap allows us to access the minimum element and insert a new element in $O(1)$, and delete elements in $O(\log N)$ amortized time. We have implemented an intrusive pairing heap, i.e., the heap-related information is stored in the edges. Intrusiveness makes the ``given an edge, remove it'' operation more natural, since there is no need for handles. The latest version of Blossom V uses pairing heaps.

    We observed that pairing heaps performed better in our algorithm, so we use them in the final version of our code.

    \item \textbf{Path compression for blossom parents:} to calculate the top node of a given node, we do not walk the entire path of blossom parents, but use path compression. Every node remembers its blossom ancestor, which is a lazily computed top node. The same is done in Blossom V. Notice that expansions may invalidate the current blossom ancestor fields for the blossom descendants of the expanded blossom.

    \item \textbf{Lazy updates for queues of internal plus-plus edges:} after the queues are updated, all queues are correct except for (possibly) the queues of plus-plus internal edges. This means that the queues except for the queues of plus-plus internal edges contain precisely the set of non-loop edges of the needed type. Queues of plus-plus internal edges are allowed to contain plus-plus loops. Whenever we need to get the minimal slack plus-plus internal non-loop edge, we keep popping the root of the heap until it is a non-loop. This makes shrink operations cheaper. The same is done in Blossom V. 

    \item \textbf{No tree children needed:} in Blossom V and the unweighted cherry-based solver from \cite{cherry_paper}, nodes in a tree remember their children. This way, the authors of \cite{blossom_v} and \cite{cherry_paper} can traverse their trees when they need to. We notice that the only case in which we need to go through all the nodes in the tree is when the tree is destroyed. Moreover, the order of the traversal does not matter. Thus, we can just make a tree remember the list of its nodes, rather than make nodes remember their tree children. This is useful with regard to shrink operations, since we do not need to define the tree children for the supernodes created. 

    \item \textbf{Removal from dead heaps:} removing an element from a pairing heap is the only operation with pairing heaps that is not done in $O(1)$ time. We observe that some edges and nodes might be in heaps that are associated with trees that have already been destroyed. In this case, we do not need to maintain this heap in a correct state. Thus, we can afford doing incorrect remove operations in a straightforward way in $O(1)$ time from such a heap. 

    An alternative way would be to explicitly destroy all the heaps associated with a tree upon the tree destruction. However, the edges and nodes from these heaps might never be included in another heap in the future, so our lazy heap destruction is more efficient.

    \item \textbf{Incident edge traversal}: we frequently need to traverse all edges incident to a node. We tried two ways of maintaining the list of neighboring non-loop edges for supernodes.

    The first way is to store the neighboring non-loop edges in a lazily computed vector. When we need to traverse the neighbors of a given node that does not have the field computed yet, we go through the tree of the blossom descendants until we reach nodes with the neighbors field computed. Then, we populate the list of neighbors with the edges from the descendants' neighbors lists, discarding loops.

    The second way is to store the neighboring non-loop edges in a linked list. Upon {\tt Shrink} operations, we steal the neighbor lists from the blossom children and splice them together to make a list of neighbors for the new node. Later, upon the first traversal, we discard loops from the list. This is done in Blossom V.

    The first approach stores the data contiguously and does less work in {\tt Shrink} operations. The second approach uses less memory. We use the first approach in the final version of our code, since it showed slightly better performance.

    \item \textbf{$E_0$ maintenance:} recall that $E_0$ is the set of zero-slack edges. During our primal phase, we need to compute a maximum matching in the graph $(V, E_0)$. This requires traversing zero-slack incident edges for a given node. One can go through all incident edges of a node, compute the true slack for every edge, and then process only the zero-slack edges. We observed a performance improvement when we added the following trick. We make the edges remember a Boolean field \texttt{maybe\_has\_zero\_slack}. If \texttt{maybe\_has\_zero\_slack=false}, then the edge has non-zero slack. Otherwise, the edge may or may not have zero slack. When we go through incident edges of a given node, we immediately discard the edges with \texttt{maybe\_has\_zero\_slack=false}. Although computing the true slack for a given edge takes $O(1)$ time, it requires accessing the endpoint nodes and the trees they belong to, which is more expensive than just accessing the edge. It turns out that maintaining the extra field \texttt{maybe\_has\_zero\_slack} improves the performance. The field is set to \texttt{true} after the dual updates for those edges that became zero-slack. However, $E_0$ is not purely incremental, since, for example, a minus-minus edge may increase its slack. For those edges that leave $E_0$, \texttt{maybe\_has\_zero\_slack} is set to \texttt{false} lazily.
\end{enumerate}

\section{Receptacle rotation lemma}\label{sec_rotation_proof}

\Cref{lemma:rotation} was presented as Lemma 3.2 in \cite{cherry_paper}.
We believe
that the argument the authors present may be incomplete. For completeness, we give our own proof below.

\begin{proof}[Proof of \cref{lemma:rotation}]
First, traverse the path $\calP^+_v$, and for each node $u$ on this path compute the distance $d(u)$ from $v$ and store it at $u$.
Define $d(u) = \infty$ for all other nodes.
Then repeat the following operation while $v\ne\r_\calB$:
\begin{itemize}
\item Let $(x,\r_\calB)$ be the last arc of $\calP^+_v$; it must be unmatched.
Traverse the path $\calP^+_x$, find a node $u$ on this path with the minimum value of $d(u)$.
We have $d(u)\le d(x)$, so $u$ lies on $\calP^+_v[v\rightarrow x]$.

The path $\calP^+_x[x\rightarrow u]$ must contain an odd number of edges; its first and last edges are matched. To see this, let $(u, w) \in \calP_v^+$ and $(u, z) \in \calP_x^+$ be the next arcs after $u$ in the paths $\calP_v^+$ and $\calP_x^+$, respectively. We have $w \neq z$, because otherwise, $\calP^+_x$ would follow $\calP^+_v$ from $u$, and $\calP^+_x$ does not go through $(x,\r_\calB)$. Moreover, both paths go through the unique matched edge incident to $u$. $\calP^+_v$ does not go through $\{u, z\}$, therefore, $\{u, w\}$ is matched, and $w$ is on the path $\calP^+_v$ with $d(w) > d(u)$. Thus, $\calP^+_x[x\rightarrow u]$ ends with a matched arc $(w, u)$. It also starts with a matched edge, since it begins with a plus-parent of $x$. Then $\calP^+_x[x\rightarrow u]$ contains an odd number of edges.

Augment the matching along the path $\{(\r_\calB,x)\} \cup \calP^+_x[x\rightarrow u]$, making $u$ an unmatched node.
Make the old $\r_\calB$ a $\pm$-node, and $u$ a purely plus-node. 
Update all plus- and minus-parents for nodes on the cycle $C=\{(\r_\calB,x)\} \cup \calP^+_x$ (except for $u$)
to point to edges on the cycle in the appropriate direction. Update $\r_\calB$ to be $u$.
\end{itemize}

After this operation, the matching remains feasible and uses only tight edges. The $\calP^+$ and $\calP^-$ paths for the nodes on the cycle $C$ end at the new receptacle by construction; $\calP^+$ and $\calP^-$ paths starting from other nodes reach $C$ and then are channeled to the new receptacle. Thus, the cherry blossom is again in a correct state. Furthermore, the path $\calP^+_v[v\rightarrow u]$ and the cycle $C$ have a single common vertex $u$, by the choice of $u$.
Thus, the path $\calP^+_v[v\rightarrow u]$ does not change, and equals $\calP^+_v$ after the update.
Distances $d(u)$ remain valid on $\calP^+_v[v\rightarrow u]$. Every iteration strictly decreases $d(\r_\calB)$, therefore the algorithm terminates after a finite number of steps.
\end{proof}

\bibliographystyle{plain}
\bibliography{lipics-v2021-sample-article}

\end{document}